\documentclass[11pt]{article}
\usepackage{multibib}
\usepackage{float}
\usepackage[export]{adjustbox}
\usepackage{graphicx}
\graphicspath{ {./images/} }
\usepackage{subcaption}
\usepackage[utf8]{inputenc}
\usepackage[export]{adjustbox}
\usepackage{wrapfig}
\usepackage{authblk}
\usepackage[top=27mm, bottom=27mm, left=22mm, right=22mm]{geometry}
\usepackage{amsthm,amsmath,amssymb}
\usepackage{mathtools}
\usepackage{mathrsfs}
\usepackage{multirow}
\usepackage{microtype}
\usepackage{hyperref}
\usepackage[capitalize]{cleveref}
\usepackage{xcolor}
\usepackage{verbatim}
\usepackage{indentfirst}
\usepackage[noadjust]{cite}
\usepackage{ulem}

\usepackage{pgfplots}
\usepgfplotslibrary{groupplots}
\pgfplotsset{compat=1.18}
\usepackage{tikz}
\usepackage{float}

\newtheorem{theorem}{Theorem}[section]
\newtheorem{lemma}[theorem]{Lemma}

\newtheorem{corollary}[theorem]{Corollary}

\newtheorem{remark}[theorem]{Remark}

\newcommand{\supp}{\operatorname{supp}}

\title{Improved Johnson-type Bounds for Insertion-Deletion Codes}

\author{
Yulin Yang\thanks{Research Center for Mathematics and Interdisciplinary Sciences, Shandong University, Qingdao 266237, China. Email:~forestyoung@mail.sdu.edu.cn.} 
}

\date{}

\begin{document}

\maketitle

\begin{abstract}
We improve upon the Johnson-type bounds of Hayashi--Yasunaga and Liu--Tjuawinata--Xing for insertion--deletion codes by encoding each local list into a binary constant-weight code. 
The resulting local list-size bound is tight over sufficiently large alphabets. 
Combining this bound with an averaging argument and the constant-weight McEliece--Rodemich--Rumsey--Welch bound yields an asymptotic rate bound that strictly improves Yasunaga's Elias-type bound throughout the nontrivial range.
\end{abstract}



\section{Introduction}

\noindent Codes designed to correct insertion and deletion errors are known as insertion--deletion (insdel) codes. 
The study of insdel codes is motivated by applications such as DNA-based data storage~\cite{cai2021correcting}, file synchronization~\cite{hanna2019guess}, and racetrack memories~\cite{chee2018coding}.
To quantify the error-correcting capability of such codes, Levenshtein~\cite{levenshtein1966binary} introduced a metric now known as the Levenshtein distance. 
For two words $x,y\in[q]^n$, their Levenshtein distance $d_{\mathrm L}(x,y)$ is defined as the minimum total number of insertions and deletions required to transform $x$ into $y$.
Consequently, a code with minimum Levenshtein distance at least $d$ can uniquely correct any combination of at most $\lfloor(d-1)/2\rfloor$ insertion and deletion errors.
A central problem in the study of insdel codes is to determine the maximum cardinality of a $q$-ary code of block length $n$ with minimum Levenshtein distance at least $d$. 
We denote this quantity by $A_q^{(\mathrm L)}(n,d)$. 

List decoding is a relaxation of unique decoding that aims to decode beyond the unique-decoding radius $\lfloor(d-1)/2\rfloor$. 
Under list decoding, there may be several codewords within the decoding radius of a received word, 
but the number of such codewords, called the list size, is required to remain small, typically polynomial in the block length $n$.

List decoding of insdel codes has been investigated in several recent works~\cite{guruswami2017deletion,liu2021efficient,guruswami2021optimally}.
Unlike in the unique-decoding setting, insertions and deletions affect list-decodability asymmetrically~\cite{haeupler2018synchronization}. 
This asymmetry requires the insertion and deletion radii to be treated separately. 
Informally, a code $C\subseteq[q]^n$ is said to be $(\tau_{\mathrm I},\tau_{\mathrm D},L)$-insdel-list-decodable if 
it can be list-decoded from up to $\tau_{\mathrm I}n$ insertions and $\tau_{\mathrm D}n$ deletions with a list of size at most $L$.

To formalize this notion, for a word $z\in[q]^N$ and nonnegative integers $s,t$, 
let $B_{\mathrm L}(z,s,t)$ denote the set of words that can be obtained from $z$ by at most $s$ insertions and at most $t$ deletions. 
For nonnegative real numbers $\tau_{\mathrm I},\tau_{\mathrm D}$ and a positive integer $L$, 
a code $C\subseteq[q]^n$ is said to be $(\tau_{\mathrm I},\tau_{\mathrm D},L)$-insdel-list-decodable if, 
for every integer $N\in[n-\tau_{\mathrm D}n, n+\tau_{\mathrm I}n]$ and every word $z\in[q]^N$, 
we have $|B_{\mathrm L}(z,\tau_{\mathrm D}n,\tau_{\mathrm I}n)\cap C|\leq L$.
Note that the two radii $\tau_{\mathrm I}$ and $\tau_{\mathrm D}$ appear in reverse order in $B_{\mathrm L}(z,\tau_{\mathrm D}n,\tau_{\mathrm I}n)$ 
because passing from the received word $z$ back to a candidate codeword interchanges the roles of insertions and deletions.

The Johnson bound serves as a benchmark in the study of list-decodability of codes. 
For a fixed alphabet, it gives a lower bound on the list-decoding radius that depends only on the relative minimum distance of the code.

\subsection{Prior work}

\paragraph{Johnson-type bounds.} The study of Johnson-type bounds for insdel codes was initiated by Wachter-Zeh~\cite{wachter2018list}. 
Hayashi and Yasunaga~\cite{hayashi2020list} subsequently made several amendments to this result and derived the following bound.

\begin{theorem}[{\cite[Theorem~1]{hayashi2020list}}]\label{thm:HY}
Let $q\geq2$ be an integer, and let $n,d$ be positive integers with $d\leq2n$. 
Let $C\subseteq [q]^n$ be a code with minimum Levenshtein distance at least $d$. 
Let $s,t$ be nonnegative integers with $s\le n$. 
If $d>2s+\frac{2t(n-s)}{n-s+t}$, then for every $N\in[n-s,n+t]$ and every $z\in [q]^{N}$, we have
\begin{equation}\label{eq:HY}
    |B_{\mathrm{L}}(z,s,t)\cap C|\leq\frac{(n+t)d}{(n-s+t)(d-2s)-2t(n-s)}.
\end{equation}
\end{theorem}

Their proof adapts the classical double-counting argument for the Johnson bound in the Hamming metric. 
For a fixed received word $z$, they consider the decoding list $\mathcal L=B_{\mathrm L}(z,s,t)\cap C$ and double-count the sum of pairwise Levenshtein distances within this list, 
$\sum_{c,c'\in\mathcal L}d_{\mathrm L}(c,c')$.

Later, Liu, Tjuawinata, and Xing~\cite{liu2023lower} derived the following lower bound on the insdel list-decoding radius, 
which improves upon the Hayashi--Yasunaga bound in certain parameter regimes.

\begin{theorem}[{\cite[Theorem~1.1]{liu2023lower}}]\label{thm:LTX}
Let $q\geq2$ be an integer, and let $n,d$ be positive integers with $d\leq2n$.
Let $C\subseteq[q]^n$ be a code with minimum Levenshtein distance at least $d=2\delta n$, and let $L\geq2$ be an integer. 
Let $\tau_{\mathrm I}$ and $\tau_{\mathrm D}$ be nonnegative real numbers satisfying $\tau_{\mathrm I}<q-1$ and $\tau_{\mathrm D}<\frac{q-1}{q}$. 
If $\tau_{\mathrm D}<\delta$ and
\begin{equation}\label{eq:LTX}
\tau_{\mathrm I}<\max_{r\in[L]}\left\{\frac{2L-r+1}{L+1}(1-\tau_{\mathrm D})-\frac{L}{r}(1-\delta)\right\},
\end{equation}
then $C$ is $(\tau_{\mathrm I},\tau_{\mathrm D},L)$-insdel-list-decodable.
\end{theorem}

Their method is to associate each codeword $c$ in the decoding list of a fixed received word $z$ with the set $Y_c\subseteq[N]$ of positions in $z$ corresponding to a longest common subsequence of $z$ and $c$. 
They then analyze unions of higher-order intersections of these sets via the inclusion--exclusion principle and take carefully chosen linear combinations of the resulting inequalities to obtain their bound.

\paragraph{Code-cardinality and rate bounds.}
As in the classical Elias--Bassalygo argument for the Hamming metric, averaging converts a uniform local list-size bound into a global bound on code cardinality. Specifically, if a code $C\subseteq[q]^n$ satisfies $|B_{\mathrm L}(z,s,t)\cap C|\leq L$ for every $z\in[q]^{n-s+t}$, then
\begin{equation}\label{eq:averaging-bound}
|C|\leq L\cdot\frac{q^{n-s+t}}{I_q(n-s,t)}, 
\end{equation}
where $I_q(m,t)=\sum_{i=0}^{t}\binom{m+t}{i}(q-1)^i$ is the number of length-$(m+t)$ supersequences of a fixed $q$-ary word of length $m$.

Applying the averaging estimate \eqref{eq:averaging-bound} to the Hayashi--Yasunaga bound \eqref{eq:HY} with $s=0$, Yasunaga~\cite{yasunaga2024improved} obtained the following Elias-type bound: 
for positive integers $q,n,d$ with $q\geq2$ and $d<2n$, and every nonnegative integer $t$ satisfying $t<nd/(2n-d)$, 
\begin{equation}\label{eq:Yasunaga}
A_q^{(\mathrm L)}(n,d)\leq\frac{(n+t)d}{(n+t)d-2nt}\cdot\frac{q^{n+t}}{I_q(n,t)}.
\end{equation}
As a consequence of \eqref{eq:Yasunaga}, for $0<\delta<1-1/q$, the rate function $R_q^{(\mathrm L)}(\delta)=\limsup\limits_{n\to\infty}\frac{1}{n}\log_qA_q^{(\mathrm L)}(n,\lfloor2\delta n\rfloor)$ satisfies
\begin{equation}\label{eq:Yasunaga-asymptotic}
    R_q^{(\mathrm L)}(\delta)
    \leq
    \frac{1-H_q(\delta)}{1-\delta},
\end{equation}
where $H_q(x)=x\log_q(q-1)-x\log_q x-(1-x)\log_q(1-x)$ denotes the $q$-ary entropy function.

In a very recent paper, Con, Doron, and Ribeiro~\cite[Theorem~1.3]{con2026insertion} obtained an improved upper bound on the binary deletion list-decoding capacity that is asymptotically tight as the deletion fraction tends to zero. For $0<\delta<1/20$, their result implies
\begin{equation}\label{eq:CDR-rate}
R_2^{(\mathrm L)}(\delta)\leq1-\delta-(\gamma(\delta)-\delta)H_2\left(\frac{\delta}{\gamma(\delta)-\delta}\right),
\end{equation}
where $\gamma(x)=\frac{1}{2}-\sqrt{\frac{H_2(x)\ln 2}{2}}$.

\subsection{Our contribution}

\noindent In this work, we strengthen both Theorem~\ref{thm:HY} and Theorem~\ref{thm:LTX} by revealing a binary constant-weight structure in the decoding lists. 
Specifically, after reducing the problem to received words of length $n-s+t$, we choose, 
for each codeword $c$ in the decoding list of $z$, a common subsequence of $z$ and $c$ of fixed length $n-s$, rather than a longest common subsequence as Liu, Tjuawinata, and Xing do. 
The indicator vectors of the corresponding position sets in $z$ form a binary constant-weight code of length $n-s+t$, weight $n-s$, and minimum Hamming distance at least $d-2s$. 

Our main result is stated as follows.

\begin{theorem}\label{thm:main-Johnson}
Let $q\geq2$ be an integer, and let $n,d$ be positive integers with $d\leq2n$.
Let $C\subseteq [q]^n$ be a code with minimum Levenshtein distance at least $d$. 
Let $s,t$ be nonnegative integers with $s\le \lfloor\frac{d-1}{2}\rfloor$. 
Then, for every $N\in[n-s,n+t]$ and every $z\in [q]^{N}$, we have
\begin{equation}\label{eq:main-Johnson}
    |B_{\mathrm{L}}(z,s,t)\cap C|\leq A^{(\mathrm{H})}(n-s+t,~ d-2s, ~t),
\end{equation}
where $A^{(\mathrm H)}(n,d,w)$ denotes the maximum cardinality of a binary code of block length $n$, minimum Hamming distance at least $d$, and constant weight $w$.
\end{theorem}

\begin{remark}
Recall that the Johnson bound for constant-weight codes \cite[Theorem~3]{johnson1962new} states that if $dn>2w(n-w)$, then
\begin{equation}\label{eq:Johnson-bound}
    A^{(\mathrm H)}(n,d,w)\leq\frac{dn}{dn-2w(n-w)}.
\end{equation}
If $d>2s+\frac{2t(n-s)}{n-s+t}$, applying \eqref{eq:Johnson-bound} to $A^{(\mathrm H)}(n-s+t,d-2s,t)$ gives
\begin{equation*}
    A^{(\mathrm{H})}(n-s+t,d-2s,t)
    \leq\frac{(n-s+t)(d-2s)}{(n-s+t)(d-2s)-2t(n-s)}
    \leq\frac{(n+t)d}{(n-s+t)(d-2s)-2t(n-s)}.
\end{equation*}
Thus, \eqref{eq:main-Johnson} strengthens the Hayashi--Yasunaga bound \eqref{eq:HY}. The second inequality is strict whenever $s>0$;
even when $s=0$, the first inequality can be strict. 
For example, when $(n,d,s,t)=(5,4,0,3)$, our list-size bound is $A^{(\mathrm H)}(8,4,3)=8$, whereas the Hayashi--Yasunaga bound is $16$. 
Moreover, our bound is attained for every $q\geq8$ by Theorem~\ref{thm:main-tightness}.
\end{remark}

Applying the averaging argument to Theorem~\ref{thm:main-Johnson} yields the following bound. 

\begin{theorem}\label{thm:main-Elias}
Let $q\geq2$ be an integer, and let $n,d$ be positive integers with $d\leq2n$.
Then, for all nonnegative integers $s,t$ with 
$s\le \lfloor\frac{d-1}{2}\rfloor$, we have
\begin{equation}\label{eq:main-Elias}
    A_q^{(\mathrm L)}(n,d)\leq A^{(\mathrm H)}(n-s+t,~d-2s,~t)\cdot\frac{q^{n-s+t}}{I_q(n-s,~t)}.
\end{equation}
\end{theorem}

\begin{remark}
Observe that if $d-2s>2t$, then $A^{(\mathrm{H})}(n-s+t,~d-2s,~t)=1$. Thus, \eqref{eq:main-Elias} implies that if $s+t\leq\lfloor\frac{d-1}{2}\rfloor$, then
    \begin{equation}\label{eq:shortened-sphere-packing}
        A_q^{(\mathrm{L})}(n,d)\leq\frac{q^{n-s+t}}{I_q(n-s,~t)}.
    \end{equation}
    In particular, 
    \begin{itemize}
        \item Taking $s=\lfloor\frac{d-1}{2}\rfloor$ and $t=0$ in \eqref{eq:shortened-sphere-packing}, we recover the Singleton bound: $A_{q}^{(\mathrm{L})}(n,d)\leq q^{n-\lceil\frac{d}{2}\rceil+1}$.
        \item Taking $s=0$ and $t=\lfloor\frac{d-1}{2}\rfloor$ in \eqref{eq:shortened-sphere-packing}, we recover the sphere-packing bound: $A_q^{(\mathrm{L})}(n,d)\leq\frac{q^{n+\lfloor\frac{d-1}{2}\rfloor}}{I_q(n,\lfloor\frac{d-1}{2}\rfloor)}$.
    \end{itemize}
\end{remark}

Finally, combining Theorem~\ref{thm:main-Elias} with the MRRW bound for binary constant-weight codes yields the following asymptotic rate bound.

\begin{theorem}\label{thm:main-asymptotic}
For every integer $q\geq2$ and every $0<\delta< 1$, we have
\begin{equation}\label{eq:main-asymptotic}
R_q^{(\mathrm L)}(\delta)\leq\inf_{\substack{0\le \sigma<\delta\\0\leq\omega\leq1-1/q}}
\frac{1-\sigma}{1-\omega}\left(1-H_q(\omega)+\frac{1}{\log_2 q}R_{\mathrm{LP}}\left(\frac{2(\delta-\sigma)(1-\omega)}{1-\sigma},~\omega\right)\right), 
\end{equation}
where 
\begin{equation}\label{eq:MRRW}
R_{\mathrm{LP}}(\delta,\omega)
=
\begin{cases}
H_2\left(\frac{1}{2}\left(1-\sqrt{1-\left(\sqrt{4\omega(1-\omega)-\delta(2-\delta)}-\delta\right)^2}\right)\right), & \delta< 2\omega(1-\omega), \\
0, & \delta\ge 2\omega(1-\omega).
\end{cases}
\end{equation}
\end{theorem}

\subsection{Comparison}
\paragraph{Johnson-type bounds.}
We first recall the following equivalent form of the Hayashi--Yasunaga bound, as reformulated by Liu, Tjuawinata, and Xing~\cite[Theorem~5.1]{liu2023lower}.
Let $C\subseteq[q]^n$ be a code with minimum Levenshtein distance at least $d=2\delta n$, where $0<\delta<1$.
Let $L$ be a positive integer, and let $\tau_{\mathrm I}$ and $\tau_{\mathrm D}$ be nonnegative real numbers satisfying $\tau_{\mathrm D}<\delta$.
If
\begin{equation}\label{eq:HY-list-decodability}
    \tau_{\mathrm I}<\frac{(L+1)(1-\tau_{\mathrm D})(\delta-\tau_{\mathrm D})-\delta}{1+L(1-\delta)},
\end{equation}
then $C$ is $(\tau_{\mathrm I},\tau_{\mathrm D},L)$-insdel-list-decodable.

We next state a list-decoding consequence of Theorem~\ref{thm:main-Johnson}. 
\begin{corollary}\label{cor:LD-radius}
Let $q\geq2$ be an integer, and let $n,d$ be positive integers with $d\leq2n$.
Let $C\subseteq[q]^n$ be a code with minimum Levenshtein distance at least $d=2\delta n$, and let $L\geq2$ be an integer.
Let $\tau_{\mathrm I}$ and $\tau_{\mathrm D}$ be nonnegative real numbers satisfying $\tau_{\mathrm D}<\delta$.
If
\begin{equation}\label{eq:LD-radius}
    \tau_{\mathrm I}<\max_{r\in[L]}\left\{\frac{L+1}{r+1}\left[\frac{2L-r+1}{L+1}(1-\tau_{\mathrm D})-\frac{L}{r}(1-\delta)\right]\right\},
\end{equation}
then $C$ is $(\tau_{\mathrm I},\tau_{\mathrm D},L)$-insdel-list-decodable.
\end{corollary}
Since the $r=L$ branch equals $\delta-\tau_{\mathrm D}>0$ and $\frac{L+1}{r+1}\geq1$, Corollary~\ref{cor:LD-radius} strengthens Theorem~\ref{thm:LTX}.

Figure~\ref{fig:list-decoding-comparison} compares our insertion-radius threshold \eqref{eq:LD-radius} with the Liu--Tjuawinata--Xing threshold \eqref{eq:LTX}, 
the Hayashi--Yasunaga threshold \eqref{eq:HY-list-decodability}, and the unique-decoding threshold $\tau_{\mathrm I}<\delta-\tau_{\mathrm D}$ for $L=2$ and $\delta=0.9$.

\begin{figure}[H]
    \centering
    \includegraphics[width=0.72\linewidth]{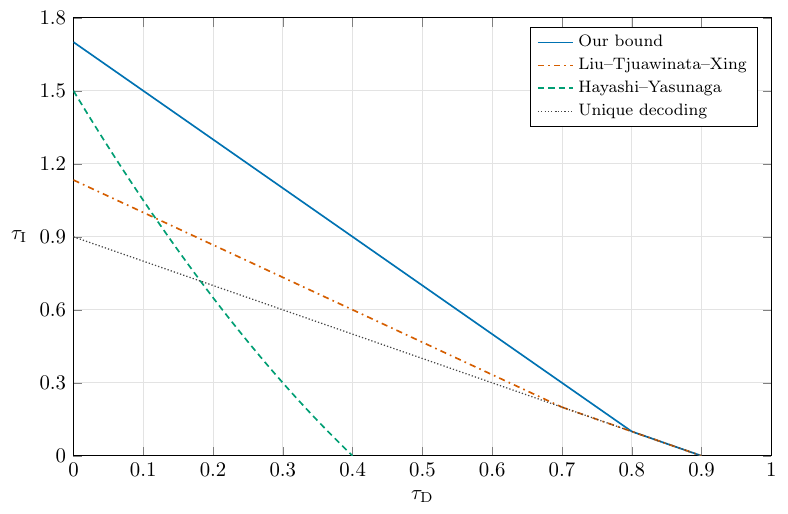}
    \caption{Comparison of Johnson-type bounds for $L=2$ and $\delta=0.9$.}
    \label{fig:list-decoding-comparison}
\end{figure}

\paragraph{Code-rate bounds.}

Setting $(\sigma,\omega)=(0,\delta)$ in the right-hand side of \eqref{eq:main-asymptotic} recovers Yasunaga's Elias-type bound \eqref{eq:Yasunaga-asymptotic}, since $R_{\mathrm{LP}}(2\delta(1-\delta),\delta)=0$. 
In fact, the improvement is strict throughout $0<\delta<1-\frac{1}{q}$. 

To see this, fix $0<\delta<1-1/q$, set $\sigma=0$, and choose $0<\varepsilon<1-1/q-\delta$. Let $\omega_\varepsilon=\delta+\varepsilon$ and $\Delta_\varepsilon=2\delta(1-\delta-\varepsilon)$. 
Writing $F_q(\omega)=(1-H_q(\omega))/(1-\omega)$, we have $F_q'(\delta)=\log_q\!\left(q\delta/(q-1)\right)/(1-\delta)^2<0$.
Moreover, $\Delta_\varepsilon<2\omega_\varepsilon(1-\omega_\varepsilon)$, and
\begin{equation*}
\sqrt{4\omega_\varepsilon(1-\omega_\varepsilon)-\Delta_\varepsilon(2-\Delta_\varepsilon)}-\Delta_\varepsilon=\frac{\varepsilon}{\delta}+O(\varepsilon^2).
\end{equation*}
Consequently, the argument of $H_2$ in $R_{\mathrm{LP}}(\Delta_\varepsilon,\omega_\varepsilon)$ is $\varepsilon^2/(4\delta^2)+O(\varepsilon^3)$, and hence
\begin{equation*}
R_{\mathrm{LP}}(\Delta_\varepsilon,\omega_\varepsilon)=\frac{\varepsilon^2}{2\delta^2}\log_2\frac{1}{\varepsilon}+O(\varepsilon^2)=o(\varepsilon).
\end{equation*}
Therefore, the value of the objective in \eqref{eq:main-asymptotic} at $(\sigma,\omega)=(0,\delta+\varepsilon)$ is $F_q(\delta)+F_q'(\delta)\varepsilon+o(\varepsilon)$, which is strictly smaller than $F_q(\delta)$ for all sufficiently small $\varepsilon>0$.

Figure~\ref{fig:rate-yasunaga} illustrates this comparison for $q=2$ and $q=4$. 

\begin{figure}[H]
    \centering
    \includegraphics[width=0.88\linewidth]{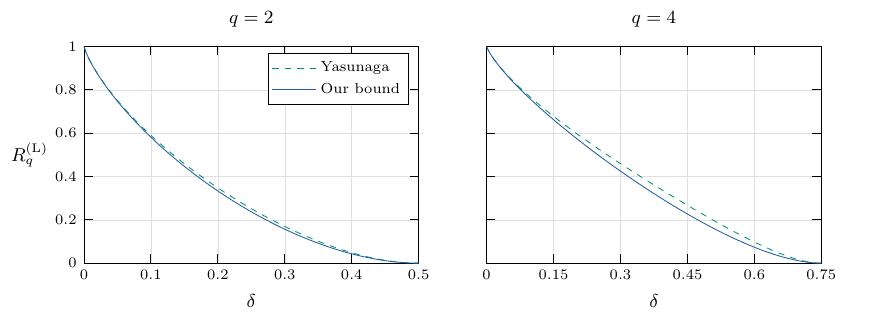}
    \caption{Comparison between our bound and Yasunaga's bound for $q=2$ and $q=4$.}
    \label{fig:rate-yasunaga}
\end{figure}

For both $q=2$ and $q=4$, numerical optimization yielded $\sigma=0$
to numerical precision at every sampled value of $\delta$. The
largest absolute reductions relative to Yasunaga's bound were
approximately $0.01468$ at $\delta\approx0.230$ and $0.04082$ at
$\delta\approx0.420$, respectively.

For binary codes in the small-distance regime $0<\delta<1/20$, Figure~\ref{fig:rate-cdr} further compares our bound with the Con--Doron--Ribeiro bound \eqref{eq:CDR-rate}. 
The two bounds in Figure~\ref{fig:rate-cdr} cross numerically at $\delta\approx0.015895$: the Con--Doron--Ribeiro bound is smaller before the crossing, whereas our bound is smaller afterward.

\begin{figure}[H]
    \centering
    \includegraphics[width=0.72\linewidth]{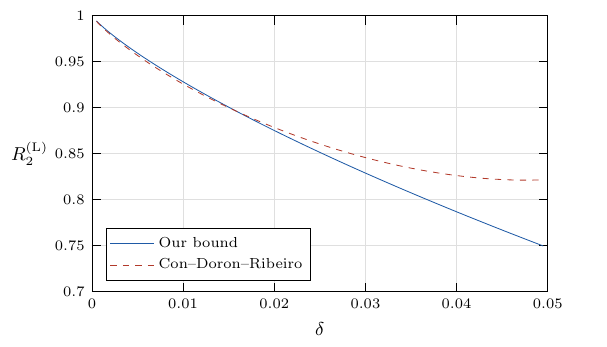}
    \caption{Comparison between our bound and the Con--Doron--Ribeiro bound for binary codes in the range $0<\delta<1/20$.}
    \label{fig:rate-cdr}
\end{figure}

\subsection{Organization}
\noindent The remainder of this paper is organized as follows. 
In Section~\ref{sec:2}, we review the basic definitions and preliminary results used throughout the paper. 
In Section~\ref{sec:3}, we prove Theorem~\ref{thm:main-Johnson} and Corollary~\ref{cor:LD-radius}. 
In Section~\ref{sec:4}, we establish Theorems~\ref{thm:main-Elias} and~\ref{thm:main-asymptotic}. 
Finally, in Section~\ref{sec:5}, we establish the large-alphabet tightness of our local bound and discuss the limitations of the resulting global bound.

\section{Preliminaries}\label{sec:2}

\noindent For two words $x,y\in[q]^n$, the Levenshtein distance $d_{\mathrm L}(x,y)$ is defined as the minimum number of insertions and deletions required to transform $x$ into $y$. 
It is well known that
\begin{equation*}
d_{\mathrm L}(x,y)=2n-2\operatorname{LCS}(x,y),
\end{equation*}
where $\operatorname{LCS}(x,y)$ denotes the length of a longest common subsequence of $x$ and $y$. 
A code $C\subseteq[q]^n$ is said to have minimum Levenshtein distance at least $d$ if $d_{\mathrm L}(x,y)\geq d$ for all distinct $x,y\in C$. 
We denote by $A_q^{(\mathrm L)}(n,d)$ the maximum cardinality of such a code.

For $x\in[q]^n$ and a nonnegative integer $t$, let $I_t(x)$ denote the set of words in $[q]^{n+t}$ obtained by inserting exactly $t$ symbols into $x$. 
Similarly, for $y\in[q]^{n+t}$, let $D_t(y)$ denote the set of words in $[q]^n$ obtained by deleting exactly $t$ symbols from $y$.
Thus, $y\in I_t(x)$ if and only if $x\in D_t(y)$.
Levenshtein~\cite{levenshtein2001efficient} showed that
\begin{equation*}
|I_t(x)|=\sum_{i=0}^t\binom{n+t}{i}(q-1)^i.
\end{equation*}
This cardinality is independent of $x$, and we denote its common value by $I_q(n,t)$.

For the asymptotic analysis, we recall the following standard entropy estimate. 
For $0\leq x\leq1$, define the $q$-ary entropy function by $H_q(x)=x\log_q(q-1)-x\log_qx-(1-x)\log_q(1-x)$, with the convention that $0\log_q0=0$. 
For every fixed $0\leq\rho\leq1-1/q$, we have
\begin{equation*}
\sum_{i=0}^{\lfloor\rho n\rfloor}\binom{n}{i}(q-1)^i=q^{H_q(\rho)n+o(n)}.
\end{equation*}

We next recall some notation and basic facts concerning binary constant-weight codes, which play a central role in our arguments.
For $x\in\{0,1\}^n$, let $\supp(x)=\{i\in[n]:x_i=1\}$ and $w_{\mathrm H}(x)=|\supp(x)|$ denote the support and Hamming weight of $x$, respectively. 
For $x,y\in\{0,1\}^n$, their Hamming distance is defined by $d_{\mathrm H}(x,y)=|\{i\in[n]:x_i\neq y_i\}|$.
We denote by $A^{(\mathrm H)}(n,d,w)$ the maximum cardinality of a binary code of length $n$, constant weight $w$, and minimum Hamming distance at least $d$.

Let $\binom{[n]}{w}=\{S\subseteq[n]:|S|=w\}$, and for $S\subseteq[n]$, let ${1}_S\in\{0,1\}^n$ denote its indicator vector. 
The correspondence $S\mapsto{1}_S$ induces a bijection between families $\mathcal F\subseteq\binom{[n]}{w}$ and binary constant-weight-$w$ codes. 
For $E,F\in\binom{[n]}{w}$, we have
\begin{equation*}
d_{\mathrm H}({1}_E,{1}_F)=2w-2|E\cap F|.
\end{equation*}
Moreover, complementation preserves Hamming distance while changing the weight from $w$ to $n-w$. 
Consequently, $A^{(\mathrm H)}(n,d,w)=A^{(\mathrm H)}(n,d,n-w)$.

For $0\leq\delta,\omega\leq1$, define the asymptotic rate of binary constant-weight codes by
$R^{(\mathrm H)}(\delta,\omega)=\limsup\limits_{n\to\infty}\frac{1}{n}\log_2 A^{(\mathrm H)}\bigl(n,\lfloor\delta n\rfloor,\lfloor\omega n\rfloor\bigr)$.
The MRRW bound for binary constant-weight codes~\cite{mceliece1977new} states that
\begin{equation*}
R^{(\mathrm H)}(\delta,\omega)\leq R_{\mathrm{LP}}(\delta,\omega).
\end{equation*}
We use the following explicit expression for $R_{\mathrm{LP}}$ from~\cite[Eq.~(10)]{ben2006upper}:
\begin{equation*}
R_{\mathrm{LP}}(\delta,\omega)
=
\begin{cases}
H_2\left(\frac{1}{2}\left(1-\sqrt{1-\left(\sqrt{4\omega(1-\omega)-\delta(2-\delta)}-\delta\right)^2}\right)\right), & \delta< 2\omega(1-\omega), \\
0, & \delta\ge 2\omega(1-\omega).
\end{cases}
\end{equation*}

\section{Johnson-type list-size bounds}\label{sec:3}

\noindent In this section, we prove our main result, Theorem~\ref{thm:main-Johnson}, and derive Corollary~\ref{cor:LD-radius}.
We begin by showing that it suffices to consider received words of the canonical length $N_0=n-s+t$.

\begin{lemma}\label{lem:length-reduction}
Let $q\geq2$ be an integer, let $C\subseteq[q]^n$ be a code, and let $s,t$ be nonnegative integers with $s\leq n$. 
Set $N_0=n-s+t$. For every integer $N\in[n-s,n+t]$ and every $z\in[q]^N$, there exists $z_0\in[q]^{N_0}$ such that
\begin{equation*}
B_{\mathrm L}(z,s,t)\cap C \subseteq B_{\mathrm L}(z_0,s,t)\cap C.
\end{equation*}
\end{lemma}

\begin{proof}
Fix a positive integer $N\in[n-s,n+t]$ and a word $z\in[q]^N$. 
If $N<N_0$, obtain $z_0$ from $z$ by inserting any $N_0-N$ symbols.
If $N>N_0$, obtain $z_0$ from $z$ by deleting any $N-N_0$ symbols.
If $N=N_0$, set $z_0=z$.

Let $c\in B_{\mathrm L}(z,s,t)\cap C$, and fix a transformation from $z$ to $c$ using $a\leq s$ insertions and $b\leq t$ deletions. 
Since $z$ and $c$ have lengths $N$ and $n$, respectively, we have $N+a-b=n$.

Suppose first that $N<N_0$, and set $k=N_0-N$. 
Starting from $z_0$, delete the $k$ symbols added in its construction and then apply the chosen transformation from $z$ to $c$. 
This uses $a\leq s$ insertions and $b+k=b+N_0-N=a-s+t\leq t$ deletions.

Suppose next that $N>N_0$, and set $k=N-N_0$. 
Starting from $z_0$, reinsert the $k$ symbols deleted in its construction and then apply the chosen transformation from $z$ to $c$. 
This uses $b\leq t$ deletions and $a+k=a+N-N_0=b+s-t\leq s$ insertions.

Finally, if $N=N_0$, then $z_0=z$, so $c\in B_{\mathrm L}(z_0,s,t)$ immediately. 
Thus, $c\in B_{\mathrm L}(z_0,s,t)$ in all cases. This completes the proof.
\end{proof}

We next bound the list size at the canonical length.

\begin{lemma}\label{lem:canonical-length}
Let $q\geq2$ be an integer, and let $n,d$ be positive integers with $d\leq2n$. 
Let $C\subseteq [q]^n$ be a code with minimum Levenshtein distance at least $d$. 
Let $s,t$ be nonnegative integers with $s\le \lfloor\frac{d-1}{2}\rfloor$. 
Then, for every $z\in [q]^{n-s+t}$, we have
\begin{equation*}
    |B_{\mathrm{L}}(z,s,t)\cap C|\leq A^{(\mathrm{H})}(n-s+t,~ d-2s, ~t).
\end{equation*}
\end{lemma}

\begin{proof}
For each $c\in B_{\mathrm{L}}(z,s,t)\cap C$, since $c$ can be obtained from $z$ by at most $s$ insertions and at most $t$ deletions, $z$ and $c$ have a common subsequence of length at least $n-s$.
Hence, we may choose a subset $S_c\subseteq[n-s+t]$ with 
$|S_c|=n-s$ such that the subsequence of $z$ indexed by $S_c$ is a subsequence of $c$.

Observe that for distinct $c,c'\in B_{\mathrm{L}}(z,s,t)\cap C$, the subsequence of $z$ indexed by 
$S_c\cap S_{c'}$ is a common subsequence of $c$ and $c'$. Thus, $|S_c\cap S_{c'}|\leq\operatorname{LCS}(c,c')$, and hence
\begin{equation*}
\begin{split}
    d_{\mathrm{H}}(1_{S_c},1_{S_{c'}})
    &=2(n-s)-2\left|S_c\cap S_{c'}\right|\\ 
    &\geq 2(n-s)-2\operatorname{LCS}(c,c') 
    = d_{\mathrm{L}}(c,c')-2s 
    \geq d-2s.
\end{split}
\end{equation*}
Therefore, $\{1_{S_c}:c\in B_{\mathrm{L}}(z,s,t)\cap C\}\subseteq\{0,1\}^{n-s+t}$ is a binary code of block length $n-s+t$, constant weight $n-s$, and minimum Hamming distance at least $d-2s$. 

Since $d>2s$, the encoding map $c\mapsto 1_{S_c}$ is injective. Hence \begin{equation*} 
|B_{\mathrm{L}}(z,s,t)\cap C| = |\{1_{S_c}:c\in B_{\mathrm{L}}(z,s,t)\cap C\}| \leq A^{(\mathrm{H})}(n-s+t,~d-2s,~n-s). 
\end{equation*} 
Observe that $A^{(\mathrm{H})}(n-s+t,~d-2s,~n-s)=A^{(\mathrm{H})}(n-s+t,~d-2s,~t)$. This completes the proof.
\end{proof}

\begin{proof}[Proof of Theorem~\ref{thm:main-Johnson}]
Fix a positive integer $N\in[n-s,n+t]$ and a word $z\in[q]^N$.
By Lemma~\ref{lem:length-reduction}, there exists
$z_0\in[q]^{n-s+t}$ such that
$B_{\mathrm L}(z,s,t)\cap C\subseteq
B_{\mathrm L}(z_0,s,t)\cap C$.
The desired bound now follows from
Lemma~\ref{lem:canonical-length}.
\end{proof}

\begin{proof}[Proof of Corollary~\ref{cor:LD-radius}]
Set $s=\lfloor\tau_{\mathrm D}n\rfloor$, $t=\lfloor\tau_{\mathrm I}n\rfloor$, and $N_0=n-s+t$.
Since insertion and deletion counts are integers, allowing at most $\tau_{\mathrm D}n$ deletions and $\tau_{\mathrm I}n$ insertions is equivalent to allowing at most $s$ deletions and $t$ insertions, respectively. 
Likewise, a positive integer $N$ belongs to $[n-\tau_{\mathrm D}n,n+\tau_{\mathrm I}n]$ if and only if it belongs to $[n-s,n+t]$.

Since $\tau_{\mathrm D}<\delta$, we have $s\leq\tau_{\mathrm D}n<\delta n=d/2$, and hence $s\leq\lfloor(d-1)/2\rfloor$. 
By Theorem~\ref{thm:main-Johnson} and the weight symmetry of constant-weight codes, it suffices to prove $A^{(\mathrm H)}(n-s+t,d-2s,n-s)\leq L$.

Suppose, to the contrary, that $A^{(\mathrm H)}(n-s+t,d-2s,n-s)\geq L+1$. 
Then there exist $S_0,\ldots,S_L\in\binom{[N_0]}{n-s}$ whose indicator vectors have pairwise Hamming distance at least $d-2s$. 
For distinct $i,j$, we have $d_{\mathrm H}({1}_{S_i},{1}_{S_j})=2(n-s)-2|S_i\cap S_j|\geq d-2s$, and therefore $|S_i\cap S_j|\leq n-d/2$.

For each $k\in[N_0]$, let $a_k=|\{i\in\{0,\ldots,L\}:k\in S_i\}|$, and let $1_{\{a_k>0\}}$ denote the indicator of the condition $a_k>0$; 
that is, it equals $1$ if $a_k>0$ and $0$ otherwise. Thus, $\left|\bigcup_{i=0}^{L}S_i\right|=\sum_{k=1}^{N_0}1_{\{a_k>0\}}$.
By double counting, we have $\sum_{k=1}^{N_0}a_k =\sum_{i=0}^{L}|S_i|=(L+1)(n-s)$ and $\sum_{k=1}^{N_0}\binom{a_k}{2}=\sum_{0\leq i<j\leq L}|S_i\cap S_j|\leq\binom{L+1}{2}\left(n-d/2\right)$.

Fix any $r\in[L]$. For every $k\in[N_0]$, we have
\begin{equation*}
1_{\{a_k>0\}}\geq\frac{2a_k}{r+1}-\frac{2\binom{a_k}{2}}{r(r+1)}.
\end{equation*}
Equality holds when $a_k=0$. 
When $a_k>0$, the difference between the two sides is $(a_k-r)(a_k-r-1)/(r(r+1))$. 
Since $a_k$ is an integer, it cannot lie strictly between the consecutive integers $r$ and $r+1$, 
so this difference is nonnegative. 
Summing over $k\in[N_0]$, we obtain
\begin{equation*}
\begin{split}
N_0\geq\left|\bigcup_{i=0}^{L}S_i\right|=\sum_{k=1}^{N_0}1_{\{a_k>0\}}
&\geq\frac{2}{r+1}\sum_{k=1}^{N_0}a_k-\frac{2}{r(r+1)}\sum_{k=1}^{N_0}\binom{a_k}{2}\\
&\geq\frac{2(L+1)}{r+1}(n-s)-\frac{L(L+1)}{r(r+1)}\left(n-\frac{d}{2}\right).
\end{split}
\end{equation*}
Using $N_0=n-s+t$, rearranging, and dividing by $n$, we obtain
\begin{equation*}
\frac{t}{n}\geq\frac{L+1}{r+1}\left[\frac{2L-r+1}{L+1}\left(1-\frac{s}{n}\right)-\frac{L}{r}(1-\delta)\right].
\end{equation*}
Since $t/n\leq\tau_{\mathrm I}$ and $s/n\leq\tau_{\mathrm D}$, it follows that
\begin{equation*}
\begin{split}
\tau_{\mathrm I}\geq\frac{t}{n}
&\geq\frac{L+1}{r+1}\left[\frac{2L-r+1}{L+1}(1-\tau_{\mathrm D})-\frac{L}{r}(1-\delta)\right].
\end{split}
\end{equation*}
Here the last inequality follows from $s/n\leq\tau_{\mathrm D}$ and $2L-r+1>0$. 
Since $r\in[L]$ was arbitrary, we conclude that
\begin{equation*}
\tau_{\mathrm I}\geq\max_{r\in[L]}\left\{\frac{L+1}{r+1}\left[\frac{2L-r+1}{L+1}(1-\tau_{\mathrm D})-\frac{L}{r}(1-\delta)\right]\right\},
\end{equation*}
contradicting \eqref{eq:LD-radius}. This completes the proof.
\end{proof}

\section{Code-cardinality and asymptotic rate bounds}\label{sec:4}

\noindent We first prove the finite-length cardinality bound by averaging the list-size estimate in Theorem~\ref{thm:main-Johnson}.

\begin{proof}[Proof of Theorem~\ref{thm:main-Elias}]
Let $C\subseteq[q]^n$ be a code with minimum Levenshtein distance at least $d$. 
For each $x\in C$, fix a subsequence $x'$ of $x$ of length $n-s$. 
Observe that the shortening map $\pi:C\to [q]^{n-s},~x\mapsto x'$ is injective; 
otherwise, two distinct codewords would share the same selected subsequence of length $n-s$, implying that their Levenshtein distance is at most $2s<d$. 

Let $y\in[q]^{n-s+t}$ be chosen uniformly at random. 
On the one hand, by linearity of expectation, 
\begin{equation}\label{eq:lower-bound}
\mathbb{E}\left|D_t(y)\cap \pi(C)\right|
=\sum_{x\in C}\Pr\left[y\in I_t(\pi(x))\right]
=|C|\cdot\frac{I_q(n-s,~t)}{q^{n-s+t}}.
\end{equation}
On the other hand, observe that for each $y\in[q]^{n-s+t}$, we have $\pi^{-1}(D_t(y)\cap \pi(C))\subseteq B_{\mathrm L}(y,s,t)\cap C$. Thus, by Theorem~\ref{thm:main-Johnson}, we have
\begin{equation}\label{eq:upper-bound}
\mathbb{E}\left|D_t(y)\cap \pi(C)\right|=\mathbb{E}\mathinner{|\pi^{-1}\left(D_t(y)\cap \pi(C)\right)|}
\leq\mathbb{E}\left|B_{\mathrm{L}}(y,s,t)\cap C\right|
\leq A^{(\mathrm{H})}(n-s+t,~d-2s,~t). 
\end{equation}
Combining \eqref{eq:lower-bound} and \eqref{eq:upper-bound} gives the desired bound. 
This completes the proof.
\end{proof}

We next derive the asymptotic consequence of Theorem~\ref{thm:main-Elias}.

\begin{lemma}\label{lem:asymptotic}
For every integer $q\geq2$ and every $0<\delta<1$, we have
\begin{equation}\label{eq:asymptotic}
R_q^{(\mathrm L)}(\delta)\leq\inf_{\substack{0\leq\sigma<\delta\\0\leq\omega\leq1-1/q}}\frac{1-\sigma}{1-\omega}\left(1-H_q(\omega)+\frac{1}{\log_2 q}R^{(\mathrm H)}\left(\frac{2(\delta-\sigma)(1-\omega)}{1-\sigma},\omega\right)\right).
\end{equation}
\end{lemma}

\begin{proof}
Fix $0\leq\sigma<\delta$ and $0\leq\omega\leq1-1/q$, and set $\Delta=\frac{2(\delta-\sigma)(1-\omega)}{1-\sigma}$.
Set $K=0$ if $\sigma=0$; otherwise, choose an integer $K\geq1$ such that $2K(1-\delta)\geq1-\sigma$. 
For all sufficiently large $n$, let $d_n=\lfloor2\delta n\rfloor$, $s_n=\lfloor\sigma n\rfloor-K$, $t_n=\lfloor\frac{\omega}{1-\omega}(n-s_n)\rfloor$, and $m_n=n-s_n+t_n$.

The definition of $t_n$ gives $t_n\leq\omega m_n<t_n+1$, so $t_n=\lfloor\omega m_n\rfloor$. 
It also follows that $m_n/n\to(1-\sigma)/(1-\omega)$.

We next account for the rounding of the distance parameter. 
If $\sigma=0$, then $m_n\leq n/(1-\omega)$, and hence $\Delta m_n\leq2\delta n$, which implies $d_n\geq\lfloor\Delta m_n\rfloor$. 
Suppose that $\sigma>0$, and write $u_n=\sigma n-\lfloor\sigma n\rfloor$ and $v_n=2\delta n-\lfloor2\delta n\rfloor$. 
Since $m_n\leq(n-s_n)/(1-\omega)$, a direct calculation gives
\begin{equation*}
(d_n-2s_n)-\Delta m_n \geq \frac{2(1-\delta)}{1-\sigma}(u_n+K)-v_n > 0,
\end{equation*}
where the final inequality follows from $u_n\geq0$, $v_n<1$, and $2K(1-\delta)\geq1-\sigma$. 
Thus, in either case, $d_n-2s_n\geq\lfloor\Delta m_n\rfloor$.

For all sufficiently large $n$, we also have $s_n\leq\lfloor(d_n-1)/2\rfloor$. 
Therefore, Theorem~\ref{thm:main-Elias} and the monotonicity of $A^{(\mathrm H)}$ in its distance parameter yield
\begin{equation*}
\begin{split}
A_q^{(\mathrm L)}(n,d_n)
&\leq A^{(\mathrm H)}(m_n,d_n-2s_n,t_n)\cdot\frac{q^{m_n}}{I_q(n-s_n,t_n)}\\
&\leq A^{(\mathrm H)}(m_n,\lfloor\Delta m_n\rfloor,\lfloor\omega m_n\rfloor)\cdot\frac{q^{m_n}}{I_q(n-s_n,t_n)}.
\end{split}
\end{equation*}

Since $t_n=\lfloor\omega m_n\rfloor$, the entropy estimate recalled in Section~\ref{sec:2} gives $\frac{1}{m_n}\log_q I_q(n-s_n,t_n)\to H_q(\omega)$.
Moreover, since $m_n\to\infty$, the definition of $R^{(\mathrm H)}(\Delta,\omega)$ shows that the normalized logarithm of the constant-weight term has limit superior at most $R^{(\mathrm H)}(\Delta,\omega)$. 
Taking base-$q$ logarithms in the preceding inequality, dividing by $n$, and using $m_n/n\to(1-\sigma)/(1-\omega)$ proves the bound in \eqref{eq:asymptotic} for the fixed pair $(\sigma,\omega)$. 
Taking the infimum over all admissible pairs completes the proof.
\end{proof}

\begin{proof}[Proof of Theorem~\ref{thm:main-asymptotic}]
This follows directly from Lemma~\ref{lem:asymptotic} and the MRRW bound recalled in Section~\ref{sec:2}.
\end{proof}

\section{Concluding Remarks}\label{sec:5}
\noindent We conclude the paper with an analysis of the tightness of the list-size bound in Theorem~\ref{thm:main-Johnson}. 
The following theorem shows that this bound is tight for sufficiently large $q$.
\begin{theorem}\label{thm:main-tightness}
Let $q,n,d$ be positive integers with $q\geq2$ and $d\leq2n$, and let $s,t$ be nonnegative integers with $s\leq\lfloor\frac{d-1}{2}\rfloor$.
If $q\geq n-s+t+A^{(\mathrm{H})}(n-s+t,~d-2s,~t)$, then there exist a code $C\subseteq [q]^n$ with minimum Levenshtein distance at least $d$ and a word $z\in [q]^{n-s+t}$ such that
\begin{equation*}
|B_{\mathrm L}(z,s,t)\cap C|=A^{(\mathrm H)}(n-s+t,~d-2s,~t).
\end{equation*}
Moreover, when $s=0$, the same conclusion holds under the weaker assumption $q\ge n+t$.
\end{theorem}

\begin{proof}[Proof of Theorem~\ref{thm:main-tightness}]
Let $\mathcal F\subseteq \binom{[n-s+t]}{n-s}$ with $|\mathcal F|=A^{(\mathrm H)}(n-s+t,~d-2s,~n-s)$ such that $d_{\mathrm H}(1_E,1_{F})\geq d-2s$ for all distinct $E,F\in\mathcal{F}$.
Since $q\geq n-s+t+|\mathcal F|$, we may choose symbols $a_i$, $i\in[n-s+t]$, and $b_F$, $F\in\mathcal F$, from the alphabet $[q]$, all of which are pairwise distinct.

Let $z=(a_1,a_2,\dots,a_{n-s+t})$. 
For each $F\in\mathcal F$, define $c_F$ to be the concatenation of $s$ copies of $b_F$ and the subsequence of $z$ indexed by $F$. 
Then $c_F$ has length $s+(n-s)=n$, and it can be obtained from $z$ by deleting the $t$ positions in $[n-s+t]\setminus F$ and inserting $s$ copies of $b_F$. Hence $c_F\in B_{\mathrm L}(z,s,t)$.

Let $C=\{c_F:F\in\mathcal F\}$. We claim that $C$ has minimum Levenshtein distance at least $d$. 
Indeed, for distinct $E,F\in\mathcal{F}$, by the choice of $c_E$ and $c_F$, we have $\operatorname{LCS}(c_E,c_F)=|E\cap F|$. Thus
\begin{equation*}
\begin{split}
d_{\mathrm L}(c_E,c_{F})
&=2n-2\operatorname{LCS}(c_E,c_F)\\
&=2n-2\left|E\cap F\right|=d_{\mathrm H}(1_E,1_F)+2s\geq d-2s+2s=d.
\end{split}
\end{equation*}
By construction, every $c_F$ lies in $B_{\mathrm L}(z,s,t)$, and the words $c_F$ are distinct. Hence
\begin{equation*}    
|B_{\mathrm L}(z,s,t)\cap C|=|C|=|\mathcal F|
=A^{(\mathrm H)}(n-s+t,~d-2s,~t).
\end{equation*}

When $s=0$, the same construction works without the symbols $b_F$, $F\in\mathcal F$: simply let $c_F$ be the subsequence of $z$ indexed by $F$. 
The pairwise distinctness of the symbols in $z$ gives $d_{\mathrm L}(c_E,c_F)=d_{\mathrm H}(1_E,1_F)\ge d$, so the same conclusion holds with $q\ge n+t$.
\end{proof}

The alphabet-size requirement in Theorem~\ref{thm:main-tightness} can be very large. When $s>0$, its threshold contains $A^{(\mathrm H)}(n-s+t,d-2s,t)$, which may be exponential in its block length. 
Thus, the theorem establishes local tightness in a very-large-alphabet regime, but does not imply tightness for fixed $q$ or even for alphabet sizes polynomial in the block length.

\paragraph{Local versus global tightness.}
Theorem~\ref{thm:main-tightness} concerns the list size around a fixed received word. Its tightness does not imply that the global code-cardinality bound obtained by averaging is also tight. 
For example, when $d$ is even, taking $s=d/2-1$ and $t=0$ in Theorem~\ref{thm:main-Elias} gives
\begin{equation*}
A_q^{(\mathrm L)}(n,d) \leq q^{n-d/2+1}.
\end{equation*}
In contrast, Kong, Tamo, and Wei~\cite{kong2025combinatorial} proved that, for fixed $n$ and even $d$ with $4\leq d\leq2n-2$,
\begin{equation*}
A_q^{(\mathrm L)}(n,d)=\frac{q^{n-d/2+1}}{\binom{n}{d/2-1}}(1+o_q(1))\qquad\text{as }q\to\infty.
\end{equation*}
Thus, in this large-alphabet regime, our averaged bound has the correct power of $q$ but misses the factor $\binom{n}{d/2-1}^{-1}$. 
This shows that even when the local list-size bound is attained, the averaging step may introduce a nontrivial loss.

\section*{Acknowledgements}

\noindent The author thanks Yubo Sun for introducing him to insertion-deletion codes, and Chong Shangguan for valuable suggestions.

\bibliographystyle{plain}
\normalem
\bibliography{ref}

\end{document}